\documentclass[11pt]{article}
\usepackage[a4paper, total={17.8cm, 260mm}]{geometry}
\usepackage{graphicx}
\usepackage[utf8]{inputenc}
\usepackage[T1]{fontenc}
\usepackage{amssymb}
\usepackage{amsthm}
\usepackage{amsmath}

\newtheorem{theorem}{Theorem}
\newtheorem{lemma}{Lemma}
\newtheorem{corollary}{Corollary}

\title{New Clocks, Optimal Line Formation and \\ Self-Replication Population Protocols}

\author{
Leszek G\k asieniec\\
{\small Department of Computer Science, University of Liverpool, 
Ashton Street, L69 38X, U.K.}
\and
Paul Spirakis\\
{\small Department of Computer Science, University of Liverpool,
Ashton Street, L69 38X, U.K.}
\and
Grzegorz Stachowiak\\
{\small Institute of Computer Science, University of Wroc{\l}aw, 
Joliot Curie 15, Wroc{\l}aw, Poland}
}

\begin{document}

\maketitle

\begin{abstract}
The model of population protocols is used to study distributed processes based on pairwise interactions between anonymous agents drawn from a large population of size $n.$
The interacting pairs of agents are chosen by the {\em random scheduler} and their states 
are amended by the predefined {\em transition function} governing the considered process.
The state space of agents is fixed (constant size) and the size $n$ is not known, 
i.e., not hard-coded in the transition function.
We assume that a population protocol starts in the predefined initial configuration of agents' states representing the input, and it concludes in an output configuration 
reflecting on the solution to the considered problem.
The~sequential time complexity of a protocol refers to the number of interactions required to stabilise this protocol in one of the final configurations.
The {\em parallel time} is defined as the sequential time divided by $n.$  

In this paper we consider a known variant of the standard population protocol model in which agents are allowed to be connected by edges, referred to as the {\em network constructor} model.
During an interaction between two agents the relevant connecting edge can be formed, maintained or eliminated by the transition function.
Since pairs of agents are chosen uniformly at random the status of each edge 
is updated every $\Theta(n^2)$ interactions in expectation which coincides with $\Theta(n)$ parallel time. This phenomenon provides a natural lower bound on the time complexity for any non-trivial 
network construction designed for this variant.
This is in contrast with the standard population protocol model in which efficient protocols operate in $O(\mbox{poly}\log n)$ parallel time.

The main focus of this paper is on efficient manipulation of linear structures including formation, self-replication and distribution (including pipelining) of {\em complex information} in the adopted model.
\begin{itemize}
    \item We propose and analyse a novel edge based phase clock counting parallel time $\Theta(n\log n)$ in the network constructor model, showing also that its leader based counterpart provides the same time guaranties in the standard population protocol model. Note that all currently known phase clocks can count parallel time not exceeding $O(\mbox{poly}\log n).$ 
    \item We prove that any spanning line formation protocol requires $\Omega(n\log n)$ parallel time if high probability guaranty is imposed. 
    We also show that the new clock enables an optimal $O(n\log n)$ parallel time spanning line construction, 
    which improves dramatically on the best currently known $O(n^2)$ parallel time protocol, solving the main open problem in the considered model~\cite{constructors}.
    \item We propose a new {\em probabilistic bubble-sort} algorithm in which random comparisons and transfers are limited to the adjacent positions in the sequence. 
    Utilising a novel potential function reasoning we show that rather surprisingly this probabilistic sorting 
    procedure requires $O(n^2)$ comparisons in expectation and whp, and is on par with its deterministic counterpart.  
    \item We propose the first population protocol allowing self-replication of a {\em strand} of an arbitrary length $k$ (carrying $k$-bit message of size independent of the state space) in parallel time $O(n(k+\log n)).$ The bit pipelining mechanism and the time complexity analysis of self-replication process mimic those used in the probabilistic bubble-sort argument. The new protocol permits also simultaneous self-replication, where $l$ copies of the strand can be created in parallel in time $O(n(k+\log n)\log l).$ 
    We also discuss application of the strand self-replication protocol to pattern matching.
\end{itemize}
All protocols are always correct and provide time guaranties with high probability 
defined as $1-n^{-\eta},$ 
for a constant $\eta>0.$ 
\end{abstract}


\maketitle


\section{Introduction}

The model of {\em population protocols} originates from the seminal paper of 
Angluin et al. 
\cite{DBLP:conf/podc/AngluinADFP04}. 
This model provides tools for the formal analysis of {\em pairwise interactions} between simple indistinguishable entities referred to as {\em agents}. The agents are equipped with limited storage, communication and computation capabilities. 
When two agents engage in a direct interaction their states are amended according to the predefined {\em transition function}. 
The weakest possible assumptions in population protocols, also adopted here, limit the state space of agents to a fixed (constant) size disallowing utilisation of the size of the population $n$ in the transition function.
In the {\em probabilistic variant} of population protocols adopted in this paper, 
in each step the {\em random scheduler} selects from the whole population an ordered pair of agents formed of the {\em initiator} and the {\em responder}, uniformly at random.
The lack of symmetry in this pair is a powerful source of random bits often used by population protocols.
In this variant, in addition to {\em state utilisation} one is also interested in the {\em time complexity} of the proposed solutions. 
In more recent work on population protocols the focus is on {\em parallel time} defined as the total number of pairwise interactions (sequential time) leading to the solution divided by the population size $n$. 
For example, a core dissemination tool in population protocols known as {\em one-way epidemic}~\cite{DBLP:journals/dc/AngluinAE08a} distributes simple (e.g., 0/1) messages to all agents in the population utilising $\Theta(n\log n)$ interactions or equivalently $\Theta(\log n)$ parallel time.   
The parallel time is meant to reflect on massive parallelism of simultaneous interactions. 
While this is a simplification \cite{DBLP:journals/corr/abs-2108-11613}, it provides a good estimation on locally observed time expressed in the number of interactions each agent was involved in throughout the computation process.

Unless stated otherwise, we assume that any protocol starts in 
the predefined {\em initial configuration} with all agents being 
in the same {\em initial state}.
A population protocol {\em terminates with success} if the whole population stabilises eventually, i.e., it arrives at and stays indefinitely in one of the {\em final configurations} of states representing the desired property of the solution.



\subsection{Network Constructors Model}

While in the standard population protocol model the population of agents remains unstructured, in the network {\em constructors} model introduced in~\cite{constructors} and adopted in this paper during an interaction between two agents the edge connecting them can be formed, maintained or eliminated by the transition function.
In this way the protocol instructs agents how to organize themselves into temporary or more definite network structures.
%
%
%

Note that since pairs of agents are chosen uniformly at random the status of any edge is updated on average every $\Theta(n^2)$ interactions which coincides with $\Theta(n)$ parallel time. 
With the exception of some relaxed expectations~\cite{polylog},
this phenomenon provides a natural lower bound on the time complexity of non-trivial network construction processes, see~\cite{constructors}. 

\vspace*{0.1cm}\noindent{\bf Model specificity}
Whenever possible we will use capital letters to denote states of the agents. 
In order to accommodate  edge connections the transition function governs the relation between triplets of the following type:
\[P + Q + S\ \longrightarrow\ P' + Q' + S' .\]
The first two terms on both sides of the rule refer to the states  $P$ and $Q$ of 
the initiator and the responder (respectively) before and $P'$ and $Q'$ after the interaction. 
The third term $S$ before and $S'$ after the interaction is a binary flag indicating the status of the connection between the two agents, where the edge presence is declared by $1$ and by $0$ the lack of it. 
The states of agents are often more complex being a combination of a fixed number of attributes.
Such states are represented as tuples. 
For such compound states we use vector representation with acute brackets $<,>$, 
where the individual attributes are separated by commas.


\vspace*{0.1cm}\noindent{\bf Probabilistic guarantees}
Let $\eta$ be a universal positive constant referring to the reliability of our protocols. We say that
an event occurs with {\em negligible} probability if it occurs with probability at most $n^{-\eta}$, and
an event occurs with high probability (whp) if it occurs with probability at least $1-n^{-\eta}$.
This estimate is of an asymptotic nature, i.e., we assume $n$ is large enough to validate the results.
Similarly, we say that an algorithm succeeds with high probability if 
it succeeds with probability at least $1-n^{-\eta}$.
When we refer to the probability of failure $p$ different to $n^{-\eta}$, 
we say directly {\em with probability at least $1-p$}.
Our protocols make heavy use of Chernoff bounds and the new tail bounds for sums of geometric random variables derived in \cite{geometric}.
We refer to these new bounds as Chernoff-Janson bounds.

We also use notation
$f(n)\sim g(n) \Leftrightarrow \frac{f(n)}{g(n)}\underset{n\to\infty}{\longrightarrow} 1$

\subsection{Our results and their significance} 
The model of population protocols gained considerably in  popularity in the last 15 years. We study here several central problems in distributed computing by focusing on the adopted variant of population protocols. These include {\em phase clocks}, a distributed synchronisation tool with good space, time accuracy, and probabilistic guarantees.
The first study of leader based $O(1)$ space phase clocks can be found in the seminal paper by Angluin {\em et al.} in~\cite{DBLP:journals/dc/AngluinAE08a}. Further extensions including junta based nested clocks counting any $\Theta(\text{poly}\log n)$ parallel time were analysed in~\cite{DBLP:journals/jacm/GasieniecS21}.
Leaderless clocks based on power of two choices principle were used in fast majority protocols~\cite{DBLP:conf/soda/AlistarhAG18},
and more recently constant resolution phase clocks propelled the optimal {\em majority} protocol~\cite{DBLP:journals/corr/abs-2106-10201}.
In this work we propose and analyse a new matching based phase clock allowing to count $\Theta(n\log n)$ parallel time.
This is the first clock confirming the conclusion of the slow leader 
election protocol based on direct duels between the remaining leader candidates. %
We also propose an edge-less variant of this clock based on the computed leader. 
This clock powers the first optimal $O(n\log n)$ parallel time spanning line construction, a~key component of universal network construction, improving dramatically on the best 
currently known $O(n^2)$ parallel time protocol, and solving the main open problem from~\cite{constructors}.
%

We also consider a probabilistic variant of the classical bubble-sort algorithm, in which any two consecutive positions in the sequence are chosen for comparison uniformly at random. We show that rather surprisingly this variant is on par with its deterministic counterpart as it requires $\Theta(n^2)$ random comparisons whp.
While this new result 
is of independent algorithmic interest, together with the edge-less clock they conceptually power the strand (line-segment carrying information) self-replication protocol studied at the end of this paper.

In a wider context, {\em self-replication} is a property of a dynamical system which allows reproduction.
Such systems are of increasing interest in biology, e.g., in the context of how life could have begun on Earth~\cite{Lincoln}, but also in computational chemistry~\cite{moulin2011dynamic}, robotics~\cite{freitas} and other fields.
In this paper, a larger chunk of information (well beyond the limited state capacity) is stored collectively 
in a strand  of agents. 
Such strands may represent strings in pattern matching, a large value, or a code in more complex distributed process. 
In such cases the replication mechanism facilitates an improved accessibility to this information.
%
%
We propose the first strand self-replication protocol allowing to reproduce a strand of non-fixed size $k$ in parallel time $O(n(k+\log n)).$ This protocol permits concurrent replication, where $l$ copies of a strand can be generated in parallel time $O(n(k+\log n)\log l).$ 
The parallelism of this protocol is utilised in efficient pattern matching in Section~\ref{patmat}.

\subsection{Related work}

One of the main tools used in this paper refers to the central problem of {\em leader election}, with the final configuration comprising a single agent 
in the {\em leader} state 
and all other agents in 
the {\em follower} state.
The leader election problem received in recent years greater attention in the context of population 
protocols. 
In particular, the results in~\cite{DBLP:conf/wdag/ChenCDS14, DBLP:conf/soda/Doty14} laid down the foundation for the proof that leader election cannot be solved in a sublinear time with agents utilising a fixed number of states~\cite{DBLP:conf/wdag/DotyS15}.
In further work~\cite{DBLP:conf/icalp/AlistarhG15}, Alistarh and Gelashvili studied the relevant upper bound, where they proposed 
a new leader election protocol stabilising in time $O(\log^3 n)$ assuming $O(\log^3 n)$ states per agent.

In a more recent work Alistarh {\em et al.}~\cite{DBLP:conf/soda/AlistarhAEGR17} considered more general trade-offs between the number of states used by the agents and the time complexity of stabilisation.
In particular, the authors delivered a separation argument distinguishing between {\em slowly stabilising} population protocols 
which utilise $o(\log\log n)$ states and {\em rapidly stabilising} protocols relying on $O(\log n)$ states per~agent. 
This result coincides with another fundamental result by Chatzigiannakis {\em et al.}~\cite{DBLP:journals/tcs/ChatzigiannakisMNPS11} 
stating that population protocols utilizing $o(\log n)$ states are limited to semi-linear predicates,
while the availability of $O(n)$ states (permitting unique identifiers) admits computation of more general symmetric predicates.
Further developments include also a protocol which elects the leader in time 
$O(\log^2 n)$ w.h.p. and in expectation utilizing $O(\log^2 n)$ states~\cite{DBLP:conf/podc/BilkeCER17}. 
The number of states was later reduced to $O(\log n)$ by
Alistarh et al. in \cite{DBLP:conf/soda/AlistarhAG18} and by Berenbrink et al. in \cite{DBLP:conf/soda/BerenbrinkKKO18} through the application of two types of synthetic coins.

In more recent work G\k asieniec and Stachowiak reduce memory utilisation to $O(\log\log n)$ while preserving the time complexity $O(\log^2 n)$ whp~\cite{DBLP:journals/jacm/GasieniecS21}.
The high probability can be traded for faster leader election in the expected parallel time $O(\log n\log\log n)$, see~\cite{DBLP:conf/spaa/GasieniecSU19}. This upper bound was recently reduced to the optimal expected time $O(\log n)$ by Berenbrink {\em et al.} in~\cite{DBLP:conf/stoc/BerenbrinkGK20}. One of the main open problems in the area is to establish whether one can elect a single leader in time $o(\log^2 n)$ whp while preserving the optimal number of states $O(\log\log n).$

%

\section{Two phase clocks and leader election}\label{s:clocks}

In order to compute the unique leader and confirm its computation we execute two protocols simultaneously.
Namely, the slow leader election protocol which concludes in parallel time $O(n\log n)$ whp, and the new (introduced below) {\em matching based phase clock} which counts parallel time $\Theta(n\log n)$ whp. 
The conclusion of leader election is confirmed via one-way epidemic when the final state (in this clock) is reached by any agent.
This leader is utilised in edge-less clock in nearly optimal computation of the line containing all agents, see Section~\ref{line}, and in self-replication of strands of information, see Section~\ref{replic}.

The transition rules for governing the slow leader election and 
the new clocks follow.

\subsection{Slow leader election}
In the initial configuration all agents are in state $L$ and the leader election protocol is driven by a single rule:
\[L+L\rightarrow L+F,
\]
where $L$ represents a leader candidate, and $F$ stands for a {\em follower} (or a {\em free}) agent. 
It is well known that this leader election protocol operates in the expected parallel time $\Theta(n)$, and in parallel time $\Theta(n\log n)$ whp.

\subsection{Matching based phase clock}
The proposed matching based clock assumes the constructors model in which the transition function recognises whether two interacting agents are connected by an edge or not, indicated by $1$ or $0$, respectively.
The agents begin in the predefined state $<start>.$
When two agents in state $<start>$ interact they get connected and they enter the {\em counting stage} 
with their counters set to $<0>$. Eventually these counters reach the maximum (exit) value $max$.
The values of the counters can go either up or down, depending on the rule used during the relevant interaction. 
Note also that the cardinality of the subpopulation of agents holding 
the smallest counter value in the population can only go down.
We prove that as long as there are agents with counter value below a fixed threshold, 
it is almost impossible for any other agent to reach the counter value $max.$ 
And this is the case for the first $\Theta(n^2\log n)$ interactions, i.e., $\Theta(n\log n)$ parallel time.
Note also that the number of agents taking part in the counting process is always even as they enter and leave this process in pairs. The counting stage sub-protocol guarantees that the counters of all agents which enter this stage reach level $max$ (denoted by state $<max>$) in time $\Theta(n\log n),$
see Theorem~\ref{t:counting}. 
And during the next interaction between the two connected agents in state $<max>$ the connection is dropped and the states are updated to $<end>$ indicating the exit from the counting stage.

The rules of the transition function used in the counting stage are as follows:

\vspace*{0.1cm}\noindent{\bf Initialisation}
\[<start>+<start>\ +\ 0\ \longrightarrow\ 
<0>+<0>+\ 1
\]

\noindent{\bf Timid counting}

\begin{itemize}
    \item 
    For all connected $i\le j$ and $i<max$
    \[<i>+<j>+\ 1\ \longrightarrow\ 
<i+1>+<i+1>+\ 1\]
    \item For all disconnected $i<j$
    \[<i>+<j>+\ 0\ \longrightarrow\ 
<i>+<i+1>+\ 0 \]
\end{itemize}

\noindent{\bf Maximum level epidemic}
\[ <max>+<i>+\ 0\ \longrightarrow\ <max>+<max>+\ 0
\]

\noindent{\bf Conclude and disconnect}
\[<max>+<max>+\ 1\ \longrightarrow \ <end>+<end>+\ 0
\]
\[<start>+<end>+\ 0\ \longrightarrow \ <end>+<end>+\ 0\ \text{\# takes care of the odd $n$ case}
\]


In the next subsection we discuss the rules of an alternative phase clock in which instead of a matching the agents use virtual edges connecting them with the computed leader.

\subsection{Leader based (edge-less) phase clock}
We allocate separate constant memory to host the states of the leader based clock. This allows to run the two clocks simultaneously and independently.
The followers in the leader based clock start with the counters set to  $<0>,$ and $L$ refers to the leader state. Note that state $<0>$ is initiated for the leader based clock as soon as the agent reaches
state $<max>$ or $<end>$ in the matching based clock.
Below we present the timid counting rules which now refer to the interactions with the leader $L$ along the virtual connections.

\vspace*{0.1cm}\noindent{\bf Timid counting}

\begin{itemize}
    \item Leader interactions, for $i<max$
\[<i>+\ L\ \longrightarrow\ <i+1>+\ L
\]

    \item Non-leader interactions, for $i<j$
\[<i>+<j>\ \longrightarrow\ <i>+<i+1>
\]   
    
\end{itemize}

One can show that the two clocks have the same asymptotic time performance, see Section~\ref{s:analysis} for the relevant detail.
Note also that the leader based clock can be used independently from any edge dependent process executed in the population simultaneously. 

\subsection{Periodic leader based (edge-less) clock}\label{plbc}
One can expand the functionality of the leader based clock to pace a series of consecutive rounds of a more complex process, with each round operating in parallel time $\Theta(n\log n).$
The extension is assumed to work in rounds formed of three consecutive stages 0, 1 and 2, where each stage is associated with a single execution (full turn) of the leader based clock.
The conclusion of each stage is announced with the help of one-way epidemic
in parallel time $O(\log n)$ whp. 
And when this happens all agents which received the announcement proceed to the next stage. 
This means that after at most $O(\log n)$ parallel time delay (caused the epidemic) all agents will run the clock in the same stage whp. 
Note also that while the signal to start the next stage remains in the population throughout the whole stage, it will be wiped out whp by the signal announcing the beginning of the stage that follows. And since we have 3 stages during each round the synchronisation of agents is guarantied whp.

\section{The clocks' analysis}\label{s:analysis}

In this section we provide the time and the probabilistic guaranties for the two phase clocks introduced in Section~\ref{s:clocks}.
We first analyse the matching based clock and later extend the reasoning to the leader based (edge-less) clock.
We prove the following theorem towards the end of this section.

\begin{theorem}\label{t:counting}
In either of the two clocks state $<$$max$$>$ is reached by any agent in parallel time $\Theta(n\log n)$ whp.
\end{theorem}

When the matching based clock starts working, it forms a matching consisting of unmatched pairs of agents.
In Lemmas \ref{formedges},\ref{Lemma2} we specify how fast this is done. 
In~Lemma \ref{l:01n} we prove that whp no counter in the population has value $max$ 
for as long as the smallest counter value is at most $max-d-2$.
The constant $d$ depends on $\eta$ and its value can be derived from the proof of Lemma \ref{Lemma3}.
Also, if $T$ is the time elapsed before the value $max$ is observed in the population for the first time,
using Lemma \ref{clockMintime} one can conclude that $T>(max-d-2)0.4n\ln n$ whp.
Using this inequality, the top value $max$ can be derived from $d$ and time $T=\Theta(n\log n)$ 
which upperbounds whp the parallel time of slow leader election process.

\begin{lemma}\label{formedges}
All edges of the matching are formed in the expected parallel time $\Theta(n)$ and whp $O(n\log n)$.
\end{lemma}

\begin{proof}
The probability of an interaction forming edge $i+1$ when $i$ edges are already formed is
$\frac{(n-2i)(n-2i-1)}{n(n-1)}$.
Thus the number of interactions separating formation of edges $i$ and $i+1$ has geometric distribution with the expected value
$\frac{n(n-1)}{(n-2i)(n-2i-1)}$.
Thus the expected number of interactions to form all edges is
$\sum_{i=1}^{n/2-1} \frac{n(n-1)}{(n-2i)(n-2i-1)},$ which is $ \Theta( n^2)$.

A sufficient condition to form all the edges is
that all possible $\binom{n}{2}$ pairs of agents are generated by the random scheduler.
The probability of not choosing a fixed pair
in first $cn^2\log n$ interactions
is $\left(1-1/\binom{n}{2}\right)^{cn^2\log n}$,
which is negligible for $c$ big enough.
Thus all edges are formed after parallel time $O(n\log n)$ whp.
\end{proof}



The following lemma refers to early interactions in the matching based clock.
\begin{lemma}\label{Lemma2}
After parallel time $0.51$ at least $\frac n 2$ agents belong to already formed edges whp.
\end{lemma}

\begin{proof}
Assume that so far exactly $i$ edges are formed. The probability that during an interaction edge $i+1$ is formed is $\frac{(n-2i)(n-2i-1)}{n(n-1)}$. So the expected number of interactions
$T_i$ of forming edge $i+1$ is $\frac{n(n-1)}{(n-2i)(n-2i-1)}$.
And in turn the expected number of interactions $T$ of forming first $n/4$ edges satisfies
\[T=T_0+\cdots+T_{n/4}=\sum_{i=0}^{n/4} \frac{n(n-1)}{(n-2i)(n-2i-1)}\sim
  n\int_0^{1/4} \frac{dx}{(1-2x)^2}=\frac{n}{2}.
\]

We can estimate the probability that $T$ exceeds $0.51n$ using
Chernoff-Janson bound (Thm.2.1 of \cite{geometric}) proving that it is negligible.
In this case we can substitute (for $n$ large enough)
\[
p_*=\frac{1}{4},
\mbox{ }
M\sim\frac{n}{2},
\mbox{ and }
\lambda=\frac{0.51}{0.5}>1.
\]
Thus we get that the expected number of interactions $T\ge 0.51n$ (parallel time $\ge 0.51$) with probability less than
\(e^{-p_*M(\lambda-1-\ln\lambda)}=
  e^{-\frac{1}{8}n(\lambda-1-\ln\lambda)},
\)
i.e., with negligible probability.

\end{proof}

As soon as the edges are formed communication
along them begins.
In order to analyze this process we define the {\em edge collector problem} in which one is asked to collect (draw) all edges of a given matching $M$ of cardinality $n'>\frac n 4.$ 
This process concludes when the random scheduler generates interactions along all edges of the matching.
In addition, one can also  infer from our proof that in fact a maximum matching is formed whp. However, we can show it only after the clock analysis.
Therefore denote the number of edges of collected matching by $n'$.
As we indicated in Lemma~\ref{Lemma2} this matching has more than $\frac n 4$ edges whp. 

\begin{lemma}\label{Lemma1}
For any cardinality $n'\in [n/4,n/2],$
the parallel time of the edge 
collector problem is $O(n\log n)$ whp.
In addition, the parallel time needed to collect
the last $0.05\cdot n$ edges (of the matching) 
is at least $0.4\cdot n \ln n$ whp.
\end{lemma}
\begin{proof}
The probability of collecting an edge in an interaction,
when $i$ edges are still 
missing
is $p_i=\frac{2i}{n(n-1)}.$
The number of interactions needed to collect this
edge is a random variable $X_i$ which has a geometric
distribution with the average $\frac{n(n-1)}{2i}$.
When $k$ edges are still to be collected, the expected number of interactions to collect extra $k-l$ edges is
\[\sum_{i=l}^k\frac{n(n-1)}{2i}=\frac{n(n-1)}{2} (H_k-H_l)\sim \frac{n(n-1)}{2} \ln\frac{k}{l}.
\]

Using the upper bound of lower tail  (Theorem 3.1) of
Chernoff-Janson bounds we show 
that this number of interactions is
at least $0.4 n(n-1)\ln n$ whp, for $k=0.05n$ and $l=n^{0.1}$.
And indeed, for $n$ large enough one can adopt
\[
p_*=p_l=\frac{2n^{0.1}}{n(n-1)},
\mbox{ }
M\sim\frac{n(n-1)}{2}\ln(0.05n^{0.9}) >0.44n(n-1)\ln n,
\mbox{ and }
\lambda=\frac{0.4}{0.44}<1.
\]
This way we get that the number of interactions smaller
than $0.4 n(n-1)\ln n$ with probability smaller than
\(e^{-p_*M(\lambda-1-\ln\lambda)}\le
  e^{-0.88n^{0.1}\ln n(\lambda-1-\ln\lambda)}
\),
i.e., with negligible probability.

The collection of edges concludes when the endpoints forming each edge interact with one another.
The probability of a missing interaction along some edge in the first $cn^2\log n$ interactions
is $\left(1-1/\binom{n}2\right)^{cn^2\log n}$,
which is negligible for $c$ large enough. 
Thus edge collection concludes whp in $O(n^2\log n)$ interactions translating to parallel time $O(n\log n)$.
\end{proof}

In our clock protocol the value of parameter $d>0$ depends on the constant $\eta$ with respect to the high probability guaranties.
We prove the existence of this parameter $d,$ for any $\eta'=\eta+3$.

\begin{lemma}\label{Lemma3}
In a parallel time period of length $n^{a},$ 
for $0<a<1,$ 
any edge in the matching
is used in less than $d$ interactions whp. 
\end{lemma}

\begin{proof}
By taking into account all possible subsets of $d$ out of $n^{1+a}$ interactions and using the union bound,
the probability that an edge is 
a subject to at least $d$ interactions in parallel time $n^{a}$ does not exceed 
\[\binom{n^{1+a}}{d}
   \left(\frac{2}{n(n-1)}\right)^d\le
  \left(\frac{2n^a}{n-1}\right)^d.
\]
and this value is smaller than $n^{-\eta'}$
is for $d$ large enough.
\end{proof}

\begin{lemma}\label{Lemma4}
In a parallel
time 
period of length $n^{a},$ for $0.1\le a<1$,
there are 
at most 
$2.1 n^a$ interactions along edges of the matching whp. 
\end{lemma}

\begin{proof}
The probability that a given interaction is a matching edge interaction
is $\frac{2n'}{n(n-1)}$.
Thus in a parallel time period of length $n^{a},$
there are expected $2 n^a\frac{n'}{n-1}\le 2n^a$ edge interactions.
By Chernoff bound the number of edge interactions is at most $2.1 n^a$ whp.
\end{proof}

For the clarity of the  presentation, depending on the context we will use
the notions of counters and  {\em levels}
interchangeably.

\begin{lemma}\label{l:01n}
Let $k$ be an integer where $k<max-d-2$.
There exists a constant $c,$ s.t., during parallel time period $(0.51,cn\log n)$ 
presence of any agent on level $i<k$ guaranties whp
a linear subpopulation of agents of size at least $0.1n$ on levels $j\le k$.
Also during this period no agent reaches level $max$ whp.
\end{lemma}

\begin{proof}
We prove validity of the lemma whp, i.e., with probability at least (wp) $1-n^{-\eta}$. Let $\eta'=\eta+3$.
The proof is done by induction on parallel time $t$.
First we show that in any time $t$ of the initial parallel time period $[0.51,n^{0.2}]$ the thesis of the lemma
holds wp $1-10tn\cdot n^{-\eta'}$.
Later we prove by induction that until any considered time $t$ the thesis of the lemma holds wp $1-10tn\cdot n^{-\eta'}$.
Note that this guaranties that the thesis holds whp, i.e., wp $1-n^{-\eta},$
for parallel time $t=O(n\log n)$.
Assume that all events in the thesis of the lemma hold before parallel time $t$.
We prove that if the thesis of the lemma holds before parallel time $t$, then
it also holds in time $t$ wp $1-10n^{-\eta'}$.
By the inductive hypothesis before parallel time $t$ or equivalently until parallel time $t-\frac 1 n$ the thesis of the lemma holds
wp $1-10(t-\frac 1 n)n\cdot n^{-\eta'}$.
In turn, we get that until parallel time $t$
the thesis of the lemma holds
wp $1-10tn\cdot n^{-\eta'}$.

We first prove the base step of induction.
As we proved in Lemma \ref{Lemma2}, 
during the initial parallel time $0.51$ at least $n/2$ agents enter the clock with state $(0)$ wp $1-n^{-\eta'}$ 
Some of these agents could also relocate to higher levels.
By Lemma \ref{Lemma4} applied to the initial time period $n^{0.2}$ there are at most $2.1n^{0.2}$ of the latter wp. $1-n^{-\eta'}$.
Thus in parallel time period $[0.51,n^{0.2}]$ level $0$ is the host of at least
$0.5n-2.1n^{0.2}>0.4n$ agents
constantly residing at this level wp $1-2n^{-\eta'}$.
Also, by Lemma \ref{Lemma3} no agent reaches level $max$ wp $1-n^{-\eta'}$.
So in parallel time $t=n^{0.2}$ the lemma holds  wp $1-3n^{-\eta'}$.
Note that $1-3n^{-\eta'}\ge 1-10tn\cdot n^{-\eta'}$ for any $t\ge 0.51$.

Now we prove the inductive step.
We observe first that during parallel time period $[t',t]$, where $t'=t-n^{0.1},$ all agents which entered the clock are at least once
on level $l\le k+1$ wp $1-n^{-\eta'}$.
And indeed during this period an agent avoids interactions with 
agents on levels $j\le k$ wp at most
\[(1-0.1/n)^{n^{1.1}}<e^{0.1n^{0.1}}
\]
Because of this and Lemma \ref{Lemma3}, during this period, any agent which entered the clock does not elevate to levels higher than $k+1+d$ wp $1-2n^{-\eta'}$.
Therefore no agent reaches level $max$ during period $[t',t]$
wp $1-2n^{-\eta'}$.

In order to prove the first thesis of the lemma we consider two cases.

\noindent{\bf Case 1:}
in this case in parallel time $t'$ there are at least $0.11n$ agents on levels not exceeding $k$.
Since by Lemma \ref{Lemma4} in parallel time period $[t',t]$ at most $2.1n^{0.1}$ such agents can increase their level wp. $1-n^{-\eta'}$.
And in turn, in parallel time $t$ there are at least $0.1n>0.11n-2.1n^{0.1}$ agents on levels $j\le k$.

\noindent{\bf Case~2:}
in this case in parallel time $t'$ the number of agents on levels at most $k$ is between $0.1n$ and $0.11n$
and the number of agents on levels below $k$ is at least $3n^{0.1}.$
Let $Y$ be the set of agents belonging to the levels above $k$ in time $t'.$ 
By Lemma \ref{Lemma4} the probability that in parallel time period $[t',t]$ the number of agents below level $k$ drops below $0.9n^{0.1} (=3n^{0.1}-2.1n^{0.1})$
is negligible, i.e., at most $n^{-\eta'}$.
Consider any set $X$ with $0.9n^{0.1}$ agents residing at levels smaller
than $k$ and estimate how many agents from set $Y$ 
interact with them.
For as long as $0.38n$ agents from $Y$ do not interact with $X$, the probability of interaction between an unused (not in contact with agents in $X$) agent in $Y$ and some agent in $X$ is at least $0.68n^{-0.9}$.
Any such interaction increases the number of agents on levels not exceeding $k$.
Consider a sequence of $n^{1.1}$ zeros and ones in which
position $\iota$ is one (1) if and only if either
\begin{itemize}
\item interaction $\iota$ is between an unused agent in $Y$ with an agent in $X$ if there are more than $0.38n$ unused agents in $Y,$
\item if this number is smaller than $0.38n$ value 1 is drawn with a fixed probability $0.68n^{-0.9}$.
\end{itemize}
By Chernoff bound the probability that this sequence has less than $0.6n^{0.2}$ ones is negligible, i.e., at most $n^{-\eta'}$.
Since $0.12n<0.11n+0.6n^{0.2}$ this sequence has less than $0.6n^{0.2}$ ones only when 
the number of agents elevated to levels not exceeding $k$ is smaller than $0.6n^{0.2}$.
Also by Lemma~\ref{Lemma4} during parallel time eriod $[t',t]$ at most $2.1n^{0.1}$ other agents may increase their level beyond $k$ wp $1-n^{-\eta'}$.
So in Case 2 the number of agents on levels not exceeding $k$ increases during period $[t',t]$
by at least $0.6n^{0.2}-2.1n^{0.1}.$

\noindent{\bf Case~3:}
assume that in parallel time $t'$ the number of agents on levels $j\le k$ is between $0.1n$ and $0.11n$
and also the number of agents on levels below $k$ is smaller than $3n^{0.1}$.
The probability of an interaction between one of such agents and 
an agent in set $Y$ of agents above level $k$
is at most $6n^{-0.9}$.
Any such interaction increases the number of agents on levels not exceeding $k$.
By Chernoff bound the probability that this number of interactions exceeds $7n^{0.2}$ in $[t',t]$ is negligible, i.e., at most $n^{-\eta'}$.
Thus in Case 3 the probability that the number of agents on levels at most $k$ exceeds $0.12n>0.11m+7n^{0.2}$
is negligible, i.e., at most $n^{-\eta'}$.

We now formulate Claim 1 that upperbounds the number of agents leaving levels $j\le k$
and Claim 2 that bounds from below the number of agents joining these levels
in Case 3.
Because wp $1-6n^{-\eta'}$  the levels $j\le k$ gain agents as a result of these two processes.
This will conclude the proof.

\noindent{\bf Claim~1:}
In Case 3 during parallel time period $[t',t]$ there are at most 
$0.26 n^{0.1}$ agents located at levels $j\le k$ which increment their level wp $1-4n^{-\eta'}$. 

And indeed, for as long as there are at most $0.12n$ agents on levels 
not greater than $k$,
the probability that such agent interacts as the initiator with a clock agent is at most $0.12/n$.
Such an interaction increments the level of this clock agent with probability at most $0.12/n$.
We prove that the probability of at least $0.13 n^{0.1}$ such incrementations
is negligible, i.e., at most $4n^{-\eta'}$.
Consider a sequence of $n^{1.1}$ zeros and ones in which
position $\iota$ is one if and only if either
\begin{itemize}
\item interaction $\iota$ increments initiator's level and there are at most $0.12n$ agents on levels 
not greater than $k$
\item if this number is greater than $0.38n$ value 1 is drawn with a fixed probability $0.12/n$.
\end{itemize}
By Chernoff bound this sequence has less than $0.13 n^{0.1}$ ones (1s) wp. $1-n^{-\eta'}$.
On the other hand we have at most $0.12n$ agents on levels at most $k$ wp $1-n^{-\eta'}$.
Thus wp $1-2n^{-\eta'}$ at most $0.13 n^{0.1}$ agents on levels not exceeding $k$ can increment
their levels in $[t',t]$ acting as initiators.
Analogously, we can prove that wp $1-2n^{-\eta'}$ at most $0.13 n^{0.1}$ agents on levels not exceeding $k$ can increment
their levels in $[t',t]$ acting as responders.
So altogether at most $0.26 n^{0.1}$ agents on levels $j\le k$ increment their levels during s$[t',t]$ wp $1-4n^{-\eta'}$.

\noindent{\bf Claim~2:}
In Case 3 during parallel time period  $[t',t]$ there are at least $0.75 n^{0.1}$
interactions between agents on levels $i<k$ and those residing on levels higher than $k$ wp $1-2n^{-\eta'}$. 

For as long as there are at most $0.12n$ agents on levels at most $k$,
at least $0.38n=n/2-0.12n$ agents are on levels higher than $k$.
The probability of interaction between
such agents and an agent on level $i<k$ is at least $0.76/n=2\cdot 0.38/n$.
Any such an interaction increases the number of agents on levels not exceeding $k$.
Consider a sequence of $n^{1.1}$ zeros and ones in which at
position $\iota$ is one (1) if and only if either
\begin{itemize}
\item there are at most $0.12n$ agents on levels not greater than $k$ and  interaction $\iota$ increases the number of such agents
\item the number of agents on levels up to $k$ is greater than $0.12n$ and value 1 is drawn with a fixed probability $0.76/n$.
\end{itemize}
By Chernoff bound this sequence has more than $0.75 n^{0.1}$ ones (1s) wp $1-n^{-\eta'}$.
On the other hand we have at most $0.12n$ agents on levels at most $k$ wp $1-n^{-\eta'}$.
Thus wp $1-2n^{-\eta'}$ at least $0.75 n^{0.1}$ agents on levels exceeding $k$ can reduce
their levels to at most $k$ during $[t',t]$ while acting as initiators.

Because of both Claims 1 and 2 after parallel time period $[t',t]$ there are
at least $0.51 n^{0.1} (=0.75 n^{0.1}- 0.24 n^{0.1})$ more agents on levels $j\le k$ than in parallel time $t'$.
This proves that in parallel time $t$ there are at least $0.1n$ agents on levels $j\le k$. 
\end{proof}

\begin{lemma}\label{clockMintime}
The parallel time in which the first agent achieves level $max$ is greater than
$(max-d-2)\cdot 0.4 n\ln n$ whp.
\end{lemma}

\begin{proof}
Let $t_k$ be the time when for the first time there are no agents available at levels lower than $k.$ By Lemma~\ref{l:01n} during period $[0.51,t_k]$,
there are at least $0.1 n$ agents on level $k$ or lower.
Let $n_k\in [n/4,n/2]$ be the number of edges in time $t_k$.
Thus between time $t_k$ and $t_{k+1}$ at least $0.1n$ agents
must increment their levels to $k+1$.
This is done by collecting (interacting via) edges adjacent to them.
By Lemma \ref{Lemma1} this takes parallel time at least $0.4 n\ln n$.
This process is repeated for $max-d-2$ levels when no agents
reach level $max$ whp.
\end{proof}

\begin{lemma}\label{clockMaxtime}
The first agent moves to level $max$ in parallel time $O(n\log n)$ whp.
\end{lemma}

\begin{proof}
The total parallel time to initiate $\lfloor n/2\rfloor$ edges is $O(n\log n)$ whp
by Lemma \ref{formedges}.
If the first agent achieves level $max$ earlier the lemma remains true.
If this is not the case,  
the parallel time $O(n\log n)$ is determined by collection of all 
$\lfloor n/2\rfloor$ edges which needs 
to be repeated $max$ times resulting also in the total parallel time $O(n\log n)$.
\end{proof}

Now we are ready to prove Theorem \ref{t:counting}. 
The thesis for matching based clock follows directly from Lemmas \ref{clockMintime} and \ref{clockMaxtime}.
The thesis for the leader based clock can be proved by a sequence of lemmas almost identical to 
Lemmas \ref{l:01n}, \ref{clockMintime} and \ref{clockMaxtime}.
In the analog of Lemma \ref{l:01n} we can take $n-2$ followers instead of $n'$ edges.
This is because Lemma~\ref{formedges} assures that the parallel time counted by the matching based clock is long enough 
to form all edges whp.
Note that $n-2$ agents are initiated at level $0$ of the leader based clock in patallel  time $O(\log n)$ whp
by the epidemic resulting in dismantling of the matching based clock. And in turn we can use the initial parallel time $O(n\log n)$ instead of $0.51$ in the analog of Lemma \ref{l:01n}.

\section{Optimal line formation}\label{line}


\begin{theorem}\label{lowerbound}
    Any spanning line formation protocol operating whp needs 
parallel time $\Omega(n\log n).$
\end{theorem}
\begin{proof}
    The final spanning line configuration must be preceded by one of the three critical configurations including A) {\em two lines}, where one of four edges could be inserted to form a line, B) {\em buffalo}, where one specific edge needs to be removed, or C) {\em unicorn}, where one of the two edges needs to be removed. 
    Alternatively, the final line configuration is obtained from cycle configuration from which one edge is removed. In such case we consider the only two cycle preceding configurations including D) {\em line}, where a unique edge need to be inserted, or E) {\em chord}, where specific chord needs to be removed, see Figure~\ref{f:configs}. Thus to stabilise in the final spanning line configuration, the protocol has to go through one of the {\em bottleneck} transitions having a choice of a fixed (at most $4$) number of edges. This limited choice forces $\Omega(n\log n)$ parallel time if we insist on high probability.  

\begin{figure}[ht]
    \centering
    \includegraphics[scale=0.4]{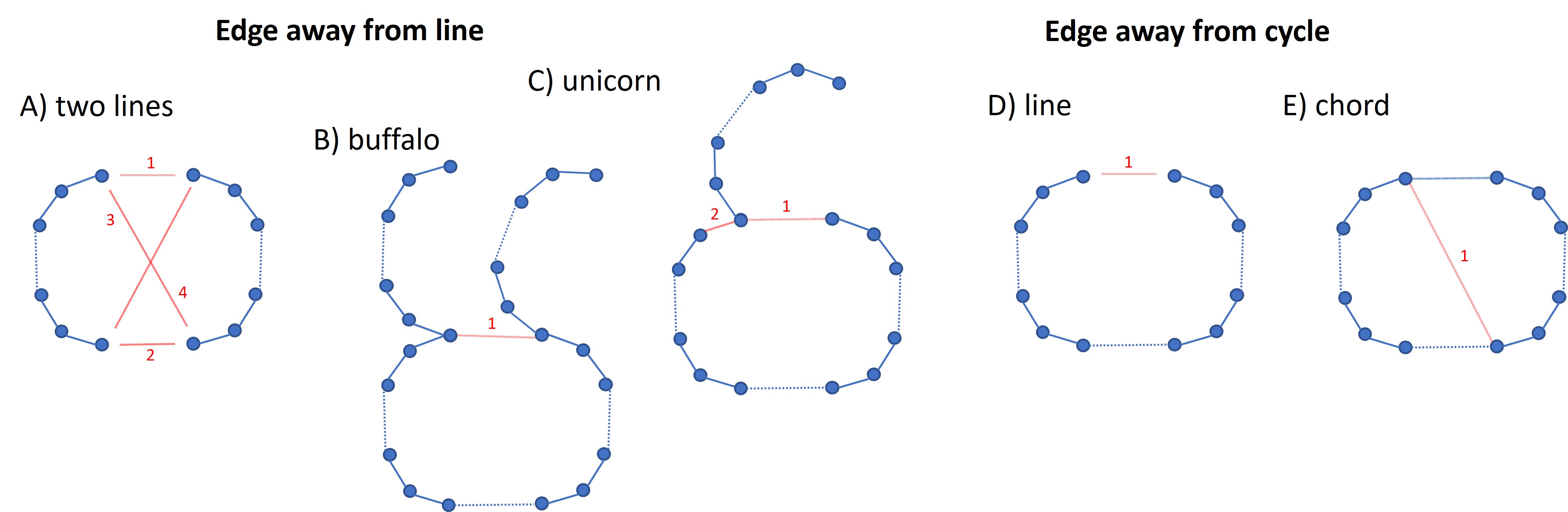}
    \caption{Configurations leading to line and cycle}
    \label{f:configs}
\end{figure}    

A more formal argument follows.
Let $T$ be the parallel time required to stabilise in the final spanning line configuration whp.
As indicated in the main part of the paper, stabilising in such configuration 
requires passing through a bottleneck transition.
At any time, when in a critical configuration,
the probability of choosing at random a pair of agents which enables a bottleneck transition is at most $4/{\binom{n}{2}}=\frac{8}{n(n-1)}$, since there are at most four such pairs for each bottleneck transition. Assume $T<\eta (n-1)\ln n/10$, where $\eta$ is the parameter of whp requirement.
The probability that the algorithm fails in this time is not smaller than the probability of no bottleneck transition which is at least
\[\left(1-\frac{8}{n(n-1)}\right)^{\eta n(n-1)\ln n/10}\sim n^{-0.8\eta}>n^{-\eta}.
\]
Thus also the probability of a spanning line formation  protocol not stabilising
in parallel time $T$ is larger than $n^{-\eta}$, for $n$ large enough.
\end{proof}

\noindent{\bf Time/space optimal line formation}
We define and analyse a new optimal line formation protocol which operates in time $\Theta(n\log n)$ whp. while utilising a constant number of extra states (not mixed with other protocols including clocks).
The protocol is preceded by leader election confirmed by the matching based clock.  
And when this happens, the periodic leader based clock starts running together with the following line formation protocol based on two main rules defined below.

\vspace*{0.1cm}\noindent{\bf Form head and tail}
\[L+F+0 \rightarrow H+T+1
\]
This rule creates the initial head in state $H$ and the tail in state $T$ of the newly formed line.
Note that since the line formation process uses separate memory the leader in the leader based clock remains in the leadership state, i.e., it is the head state $H$ is used solely in the line formation protocol.

\vspace*{0.1cm}\noindent{\bf Extend the line}
\[H+F+0 \rightarrow R+H+1
\]
This rule extends the current line by addition of an extra agent from the head end of the line. The state $R$ indicates that the agent is in the line between the head and the tail.

\begin{theorem}
The spanning line formation protocol stabilises whp  in parallel time $O(n\log n)$.
\end{theorem}
\begin{proof}
The probability of an interaction adding agent $i+1$ to the line when $i$ agents are already present is
$\frac{2(n-i)}{n(n-1)}$.
Such interactions has geometric distribution with the expected value
$\frac{n(n-1)}{2(n-i)}$.
Thus the expected parallel time of forming the line is
\[\frac{1}{n}\sum_{i=1}^n \frac{n(n-1)}{2(n-i)}\sim 
  \frac{n}{2}\sum_{i=1}^n \frac{1}{i}\sim \frac{n\ln n}{2}.
\]
By Chernoff-Janson bound this parallel time $O(n\log n)$ is guarantied whp.
\end{proof}

In order to make the line formation protocol always correct we need some backup rules for the unlikely case of desynchronisation when two or more leaders survive to the line formation stage. In such case we need to continue leader elimination.
\[L+L+0\ \longrightarrow\ L+F+0
\]
Also when a leader meets already formed head. 
\[L+H+0 \rightarrow F+H+0
\]

Finally we have to dismantle excessive lines if two or more lines are formed. This is done using extra state $D$ which dismantles the line edge by edge starting from the head.
\[H+H+0 \rightarrow H+D+0
\]
\[D+R+1 \rightarrow F+D+0
\]
\[D+T+1 \rightarrow F+F+0
\]



\section{Probabilistic bubble-sort\label{bubble-sort}}

Let array $A[0..n-1]$ contain an arbitrary sequence of $n$ numbers. 
In the {\em probabilistic bubble-sort} during each comparison step an index $j\in\{0,\dots,n-2\}$ is chosen uniformly at random, and if $A[j]>A[j+1]$ these two values are swapped in $A$.
We show that the expected number of comparisons required to sort all numbers in $A$ (in the increasing order) is $\Theta(n^2)$ whp.

In order to prove this result we first remind the reader that any sorting procedure based on fixing local inversions requires $\Omega(n^2)$ comparisons. In order to prove the upper bound we utilise the
classical {\em zero-one principle} stating that 
if a (probabilistic) sorting network sorts correctly all sequences of zeros and ones,
it also sorts an arbitrary sequence of numbers of the same length.
More precisely, if we want to prove that a given sequence $X$ of $n$ numbers will be sorted 
we have to consider only  $n+1$ zero-one sequences obtained 
by replacing $k$ largest elements of $X$ by ones and the remaining elements by zeros, for any $k=0,\dots,n$, see~\cite{Knuth73}.
Thus it is enough to prove that 
the probabilistic bubble-sort utilises $O(n^2)$ comparisons
to sort whp any zero-one sequence of length $n$,
and later use the union bound to extend this result to 
any sequence of numbers, also whp.


\begin{theorem}\label{t:ones}
The probabilistic bubble-sort 
utilises $4(n-1)(n\ln 2 +\eta\ln n)$ comparisons whp to sort 
any zero-one sequence of size $n.$
\end{theorem}

Let $k$ be the number of ones in a zero-one sequence represented by $A$.
We define a {\em configuration} $C$
as the subset of all positions in $A$ 
at which ones are situated, where $|C|=k$.
The probabilistic bubble-sort starts in the initial configuration (based on the original zero-one sequence) and thanks to the conditional swaps progresses through consecutive configurations including the final one in which all zeros precede $k$ ones.
For any configuration $C,$ 
we define a {\em potential function} $P(C)=\sum_{i\in C} P[i]$ with a non-negative integer value, 
where
\[P[i]=2^{n-k+2l-i}-2^l\mbox{ , for } l=| C\cap\{0,\dots, i-1\}|.
\]


Note that the value of this potential is zero 
for all $i$
if and only if the sequence is sorted. Thus $P(C)=0$ for a sorted sequence $C$.
Also, when all ones precede all zeros, the potential $P(C)$ is the highest possible. One can notice that always $P(C)<2^n$.

We prove the following lemma.

\begin{lemma}\label{bubble}
Let $C$ be an arbitrary configuration in $A$ and
$EP(C')$ be the expected potential of the next configuration $C'$
in the probabilistic bubble-sort.
The following inequality holds.
\[EP(C')\le \left(1-\frac{1}{4(n-1)}\right) P(C). 
\]
\end{lemma}

\begin{proof}
We split configuration $C$ 
into disjoint blocks
of indices $B_1,B_2,\ldots$, each corresponding to a solid run of ones.
For any block $B=\{x,\dots,y\}$
we define a potential 
$P(B)=\sum_{i=x}^y P[i]$.
In the subsequent configuration $C',$ 
let $EP(B')$ be the expected potential of $B'\subset C'$ based on the ones originating from $B$ in the preceding configuration $C$. 
We show that 
\[ EP(B')\le \left(1-\frac{1}{4(n-1)}\right) P(B). 
\]

Let $l=|C\cap\{0,\dots, y-1\}|$.
We have
\[P(B)= \sum_{i=x}^{y} P[i] = \sum_{i=x}^{y}2^{n-k+2(l+i-y)-i}- \sum_{i=x}^{y} 2^{l+i-y}< 2^{n-k+2l-y+1}
\]


Assume first that $y=n-1$. The inequality follows from the fact that $P(B)=P(B')=0 $
as ones located at positions in $B$ cannot be moved any further.
Thus we can assume that $y<n-1$.
Now, as either $B'=B$ or $B'=\{x,\dots,y-1\}\cup\{y+1\}$ 
and the latter happens with probability $\frac1{n-1},$ we get
\[P(B')= \sum_{i=x}^{y-1}2^{n-k+2(l+i-y)-i}- \sum_{i=x}^{y} 2^{l+i-y}+2^{n-k+2l-y-1}=
\]
\[=P(B)-2^{n-k+2l-y-1}\le \frac34 P(B).
\]

And in turn
\[EP(B')\le \left(1-\frac{1}{n-1}\right)P(B)+\frac1{n-1}\cdot\frac{3}{4} P(B)=\left(1-\frac{1}{4(n-1)}\right) P(B). 
\]

Note that any configuration $C$ is the union of disjoint blocks $B_i$ and
$P(C)=\sum_i P(B_i)$, thus also

\[EP(C')=\sum_{i} EP(B_i') \le \sum_i \left(1-\frac1{4(n-1)}\right)P(B_i) = \left(1-\frac1{4(n-1)}\right)P(C)
\]

\end{proof}

The initial value of $P(C_0)$ is bounded by $2^n$.
When after $t$ random comparisons $EP(C_t)\le n^{-\eta}$ 
the sequence is sorted whp.
This holds because the probability that after $t$ random comparisons the 
sequence is not sorted is equal to the probability that the potential is greater than zero (i.e., at least 1 as the potential is always integral). This probability is less than or equal to $EP(C_t)$ which is not bigger than $n^{-\eta}$.

Let $c = \left(1-\frac{1}{4(n-1)}\right)$.
Let also $ P(C_j)$ and $P(C_{j+1})$ be the potentials of the configurations separated by the $j$th consecutive comparison.
We have shown earlier that $EP(C_{j+1})$, conditioned on the value of $P(C_j)$, is at most $c\cdot P(C_j)$. 
This implies that the unconditional value of $EP(C_{j+1})$ is at most $c\cdot EP(C_j)$.
Thus by an induction argument it follows that after $t$ random comparisons
$EP(C_t)$ is at most ${c^t}\cdot EP(C_0)$.
Finally as $EP(C_0) = P(C_0),$ where $C_0$ is the initial configuration and its potential is not a random variable, in order to estimate $t$ we get
inequality
\[EP(C_t)\le \left(1-\frac{1}{4(n-1)}\right)^t P(C_0)\le \exp\left(-\frac{t}{4(n-1)}\right)2^n\le n^{-\eta}, 
\]
which holds for $n\ln 2+\eta\ln n\le \frac{t}{4(n-1)},$
equivalent with
\[t\ge 4(n-1)(n\ln 2 +\eta\ln n).
\]
This concludes the proof of Theorem~\ref{t:ones}.

\section{Strand self-replication}\label{replic}

In this section we propose and analyse the first self-replication mechanism allowing efficient concurrent reproduction of many, possibly different, {\em strands} (line-segments) of agents carrying non-fixed size $0/1$ bit-strings.
Note that replication of fixed-size bit-strings is trivial as they
can be encoded in the state space utilised by agents.
The front agent in the strand is called {\em the head}, the last one is called {\em the tail}, and the remaining ones are called {\em regular} or {\em internal} agents.
%

The strand self-replication protocol mimics the pipelining process utilised 
and analysed in the probabilistic bubble-sort algorithm in Section~\ref{bubble-sort}.
There are, however, some small differences between the two processes.
In particular, during strand self-replication, 
the transfer (pipelining) of consecutive bits of information 
between the old and the new strand is done 
at the same time as the new strand is being constructed.
Also any bit transferred to the new strand stops moving as soon as
it finds the first unused (newly added) agent in the new strand.
Finally, the probability of using an edge in the strand is 
${1}/{\binom{n}{2}}$ comparing to the probability ${1}/{k}$
of choosing any pair of numbers in the probabilistic bubble-sort applied to sequence of size $k$.
In the proof of Theorem~\ref{replicate} we point out that only the last difference 
separates the self-replication process by a multiplicative factor $\Theta(n^2/k)$ from 
bubble-sort applied to a sequence with $k$ ones on the left and $2k$ zeros at the right end. 

{\bf The Algorithm}
When a strand is ready for self-replication it first creates a copy of its head, then pushes through this new head (one by one, preserving the order) the bits of information pulled from its own agents.
At the same time, in order to accept the incoming bits of information, first the new head and later the copies of the consecutive 
regular agents of the replicated strand extend the new strand until the tail is formed. 
When the last (tail) bit of information is delivered to the new strand, the edge bond bridging the two heads is removed and the original (old) strand is ready for the next round of self-replication. 
Note that in this version of the self-replication protocol a newly formed strand may simultaneously seek its tail extension and already be involved in self-replication from its head end.  

\begin{theorem}\label{replicate}
The strand self-replication protocol recreates a $k$-bit strand 
in parallel time $O(n(k+\log n))$ whp. 
\end{theorem}
\begin{proof}
We start with presenting further detail of the self-replication protocol.
As in all other protocols studied in this paper, the agents utilise a constant 
number of states, this time organised into the following triplets
$$<Role, B_i, Buffer>,$$
where the three attributes are:
\begin{description}
\item[Role] 
refers to the strand's 
{\em head} agent $H$, 
the {\em tail} agent $T$, or to {\em a regular} contributor $R$.

\item[B$_i$] refers to 
the bit of information combined with its position $i$ on the line. Note that each position $i$ is computed (and stored) modulo $3$ counting from the head's position $0.$ 
This allows agents to distinguish between the two directions: towards the head and towards the tail on the strand.
We use $B_T$ to denote the bit located in the tail agent of the strand.
Finally, by $|B_i|$ we denote the sole value of the information bit without its location. I.e., $B_i=<|B_i|,i>.$
Thus the binary representation of the information stored in the replicated strand corresponds to the sequence 
$|B_0|,|B_1|,\dots,|B_T|$.

\item[Buffer] is a part of agents' memory handling single information $|B_i|$ bits or control messages.
An agent is in the neutral state $\phi$ when its buffer is empty and no dedicated replication action (apart from waiting for further instructions) is needed from this agent. 
In the {\em self-replicated strand}
state $\phi^H$ denotes the empty buffer of an agent supporting bit transfer towards $H.$
Similarly, when the buffer is occupied by a bit $|B_x|$ moving towards $H$ the relevant state is $|B_x|^H$.

In the {\em newly formed strand}
state $\psi^T$ denotes the empty buffer of an agent supporting bit transfer towards the tail end.
And similarly, when the buffer is occupied by a bit $|B_x|$ moving towards the tail end the relevant state is $|B_x|^T$.
Here the control message $\psi^N$ indicates that further extension is expected at the current end of the new strand.
In this strand we distinguish also state $\psi$ (await further information) in agents just added at the tail end.
\end{description}

\noindent
Below we explain how the information (the sequence of bits) is transferred from the old to the new strand.
%
The full list of strand self-replication protocol rules follows.
Please note that this set of rules is designed for strands containing 
at least three agents, i.e., when all type of agents $H,T$ and $R$ are used.
The relevant protocols for shorter strands are trivial as they carry 
a constant number of bits. 

\vspace*{0.1cm}\noindent{\bf }
{\bf (R1) Start of the strand self-replication}
The replication process begins when the head $H$ in the neutral state $\phi$ interacts with a free agent in state $F.$ 

$$<H,B_0,\phi> +\ F\ +\ 0
\ \longrightarrow$$
$$<H,B_0,\phi^H>+<H,B_0,\psi>+\ 1$$

When this rule is applied, in the old strand signal $\phi^H$ (move all bits towards the head) is created, and in the new line signal $\psi$ means await further instructions (either to add a new agent or to conclude the replication process).

\vspace*{0.1cm}\noindent{\bf }

\vspace*{0.1cm}\noindent{\bf }
{\bf (R2) Create $|B_i|^H$ or $|B_T|^H$ bit message}
When signal $\phi^H$ arrives at the ${(i-1)}^{th}$ agent and the $i^{th}$ agent is neutral, message $|B_i|^H$ is placed in the buffer of the latter.

$$<R,B_i,\phi>+<R|H,B_{i-1},\phi^H>+\ 1
\ \longrightarrow$$
$$<R,B_i,|B_i|^H>+<R|H,B_{i-1},\phi^H>+\ 1$$

A similar action is taken at the tail agent in neutral state $<T,B_T,\phi>$

$$<T,B_T,\phi>+<R|H,B_{i-1},\phi^H>+\ 1
\ \longrightarrow$$
$$<T,B_T,|B_T|^H>+<R|H,B_{i-1},\phi^H>+\ 1$$

\vspace*{0.1cm}\noindent
The rules in {\bf R2} enable propagation of the request to pipeline all information bits towards the head $H.$ 
The rules {\bf R3} and {\bf R4} govern the relevant bit movement.

\vspace*{0.1cm}\noindent{\bf }
{\bf (R3) Move a non-tail bit message $|B_x|^H$ towards $H$}

$$<R,B_i,|B_x|^H>+<R|H,B_{i-1},\phi^H>+\ 1
\ \longrightarrow$$
$$<R,B_i,\phi^H>+<R|H,B_{i-1},|B_x|^H>\ +1$$

Note that when the bit message $|B_x|$ is moved state $\phi^H$ requesting further bit messages remains in the $i^{th}$ agent.

\vspace*{0.1cm}\noindent{\bf }
{\bf (R4) Move the tail bit message $|B_T|^H$  towards $H$}

$$<T|R,B_i,|B_T|^H>+<R|H,B_{i-1},\phi^H>+\ 1
\ \longrightarrow$$
$$<T|R,B_i,\phi>+<R|H,B_{i-1},|B_T|^H>+\ 1$$

Note that when the tail message $|B_T|^H$ is moved the neutrality of the tail agent is restored.
Eventually, thanks to the final transfer of the tail message (to the new strand) states of all buffers in the old strand are reset to $\phi$. 

\vspace*{0.1cm}\noindent
The following two rules govern transfer of bit messages between the old and the new strand.

\vspace*{0.1cm}\noindent{\bf }
{\bf (R5) Transfer a non-tail bit message $|B_x|^H$ to the head of the new strand}

$$<H,B_0,|B_x|^H>+<H,B_{0},\psi^T>+\ 1
\ \longrightarrow$$
$$<H,B_0,\phi^H>+<H,B_{0},|B_x|^T>+\ 1$$

After the transfer across the two strands the bit message is now targeting the tail end.

\vspace*{0.1cm}\noindent{\bf }
{\bf (R6) Transfer the tail bit message $|B_T|^H$ to the head of the new strand}

$$<H,B_0,|B_T|^H>+<H,B_{0},\psi^T>+\ 1
\ \longrightarrow$$
$$<H,B_0,\phi>+<H,B_{0},|B_T|^T>+\ 0$$

As indicated earlier, transfer of the tail message to the new strand and removal of the bridging edge restore the neutrality of the old strand which is now ready to reproduce again.

\vspace*{0.1cm}\noindent
Finally, we discuss the remaining rules governing the new strand creation.
Recall that the control message represented by state $\psi$ at the current end of the new strand indicates that this strand can be still extended. 

\vspace*{0.1cm}\noindent{\bf }
{\bf (R7) Move a non-tail message $|B_x|^T$ towards the tail end}

$$<H|R,B_i,|B_x|^T>+<R,B_{i+1},\psi^T>+\ 1
\ \longrightarrow$$
$$<H|R,B_i,\psi^T>+<R,B_{i+1},|B_x|^T>+\ 1$$

After this move the $i^{th}$ agent in the new strand awaits further bit messages.

\vspace*{0.1cm}\noindent{\bf }
{\bf (R8) Move the tail message $|B_T|^T$ towards the tail end}

$$<H|R,B_i,|B_T|^T>+<R,B_{i+1},\psi^T>+\ 1
\ \longrightarrow$$
$$<H|R,B_i,\phi>+<R,B_{i+1},|B_T|^T>+\ 1>$$

After this move the neutrality of the $i^{th}$ agent in the new strand is restored, i.e., no further bit messages from the head end are expected. 

\vspace*{0.1cm}\noindent
When there is no room for the bit message coming from the head end another agent has to be added to the tail end of the new strand. This is done in two steps. In the first step the current tail end requests addition of a new agent with control message $\psi^N.$ 

\vspace*{0.1cm}\noindent{\bf }
{\bf (R9) Request strand extension with $\psi^N$ on non-tail bit message $|B_x|^T$ arrival}

$$<R,B_i,|B_x|^T>+<R,B_{i+1},\psi>+\ 1
\ \longrightarrow$$
$$<R,B_i,|B_x|^T>+<R,B_{i+1},\psi^N>+\ 1.$$

The analogous rule requesting extension beyond the head of the new strand is

$$<H,B_0,|B_1|^H>+<H,B_{0},\psi>+\ 1
\ \longrightarrow$$
$$<H,B_0,|B_1|^H>+<H,B_{0},\psi^N>+\ 1.$$

\noindent
When ready (signal $\psi^N$ is present) the new agent is added from the pool of free agents.

\vspace*{0.1cm}\noindent{\bf }
{\bf (R10) Extend the new strand}

$$<H|R,B_{i},\psi^N>+\ F\ +\ 0
\ \longrightarrow$$
$$<H|R,B_{i},\psi^T>+<R,*,\psi>+\ 1$$

Note that after this rule is applied the newly added agent still awaits its bit message which is denoted by $*$.
This new bit message arrives with the help of the following two rules.

\vspace*{0.1cm}\noindent{\bf }
{\bf (R11) Arrival of a non-tail bit message}

$$<H|R,B_{i},|B_x|^T>+<R,*,\psi>+\ 1
\ \longrightarrow$$
$$<H|R,B_{i},\psi^T>+<R,B_x,\psi>+\ 1$$

As a non-tail bit arrived the new strand will be still extended which is denoted by messages $\psi^T$ (expect more bit messages from the head end)
in the $i^{th}$ agent and 
$\psi$ (further extension still possible).
The situation is different when the tail bit message arrives.

\vspace*{0.1cm}\noindent{\bf }
{\bf (R12) Arrival of the tail bit message}

$$<R,B_{i},|B_T|^T>+<R,*,\psi>+\ 1
\ \longrightarrow$$
$$<R,B_{i},\phi>+<T,B_T,\phi>+\ 1$$

After this rule is applied the neutrality at the tail end of the new strand is restored. 

Note, however, that since the neutrality of the agents closer to the head of this line was restored earlier the front of the new line can be already involved in the next line replication process. But since we use different messages for the transfers in the old and the new lines, the two simultaneously run processes will not interrupt one another.

We conclude the proof of Theorem~\ref{replicate} with Lemma~\ref{correct} stating the correctness of the proposed self-replication protocol, and Lemma~\ref{comp} addressing the parallel time complexity.

\begin{lemma}\label{correct}
The strand self-replication protocol based on rules {\bf R1-R12} is correct.
\end{lemma}
\begin{proof}
We argue first about correctness of the proposed protocol in the replicated (old) strand.
One can observe that the bit messages stored in the agents of the strand move along consecutive edges towards the head $H.$ They do not change their order as they only move when the preceding bit message vacates the relevant buffer.  
Finally, to conclude the replication process neutrality of each agent need to be restored, and this is done by the eventual transfer of the tail message $|B_T|^H.$
In what follows we discuss the actions in all three types of agents in the strand.

\begin{itemize}
    \item The actions of the tail node are governed by rules {\bf R2} and {\bf R4}. The first rule creates bit message $|B_T|^H$ and the second moves this message towards the head of the line, restoring the neutrality of the tail agent.
    \item The actions of a regular node require also rule {\bf R3} which supports movement of multiple non-tail bit messages towards $H.$
    And when the tail bit message arrives the neutrality of this regular agent is restored by rule {\bf R4} applied to this agent twice, first on the right then on the left side of this rule.

    \item The actions of head $H$ are more complex. The self-replication begins with application of rule {\bf R1} which comprises three different actions: forming a bridging edge, adding the head of the new line, and replication of its bit message in the newly formed head.
    This is followed by
    transfer of non-tail bit messages to the new strand by alternating use of rules {\bf R3} and {\bf R5}.
    When eventually the tail bit message arrives during application of rule {\bf R4}, the neutrality of the head is restored by rule {\bf R6}.
    This concludes the replication process.
\end{itemize}

For the full cycles of rules utilised in the replicated strand see Figure~\ref{f:old}.

\noindent
The new line formation requires different organisation of states and transitions. Note that all agents added to the line must originate in state $F,$ see Figure~\ref{f:new}.
Also in this case we argue that the bit messages arrive in the unchanged order and eventually the neutrality of all agents is reached (starting from the head and finishing with the tail agent) with the help of the tail bit message $|B_T|^T$.

\begin{itemize}
    \item Formation of the tail agent requires application of only two rules: {\bf R10} to add a new agent and {\bf R12} to equip this agents with the tail message $|B_T|,$
    when neutrality of this agent is reached.
    \item The situation with the regular nodes is more complex as they have to accept and store their own bit message $|B_i|$ (done by rule {\bf R11}), add additional agent (via alternating application of rules {\bf R9} and {\bf R10}) moving all non-tail bit messages following $|B_i|$ in the old strand (rule {\bf R7}) until the tail bit message arrives (rule {\bf R8}) and finally neutrality of the regular agents is reached via rule {\bf R8} or rule {\bf R12} if the agent precedes the tail agent.
    \item Rule {\bf R1} creates the head of the new line, rules {\bf R9} and {\bf R10} add a new agent, rules {\bf R5} and {\bf R7} move non-tail bit messages in the direction of the tail until the tail bit message arrives (rule {\bf R6}) when the neutrality of the head is reached (rule {\bf R8}). 
\end{itemize}

For the full cycle of rules used by agents in the replicated strand see Figure~\ref{f:new}.

\noindent
As discussed earlier in the new strand what matters is that neutrality is reached earlier by agents located closer to the head, as this strand is allowed to start self-replication while some bit messages (from the old strand) are still being moved towards the tail end (which may not be fully formed yet). However, it is enough to observe that these two replication processes are independent as they are based on movement of bit messages towards the opposite directions and in turn they share no rules.

\begin{figure}[ht]
    \centering
    \includegraphics[scale=0.35]{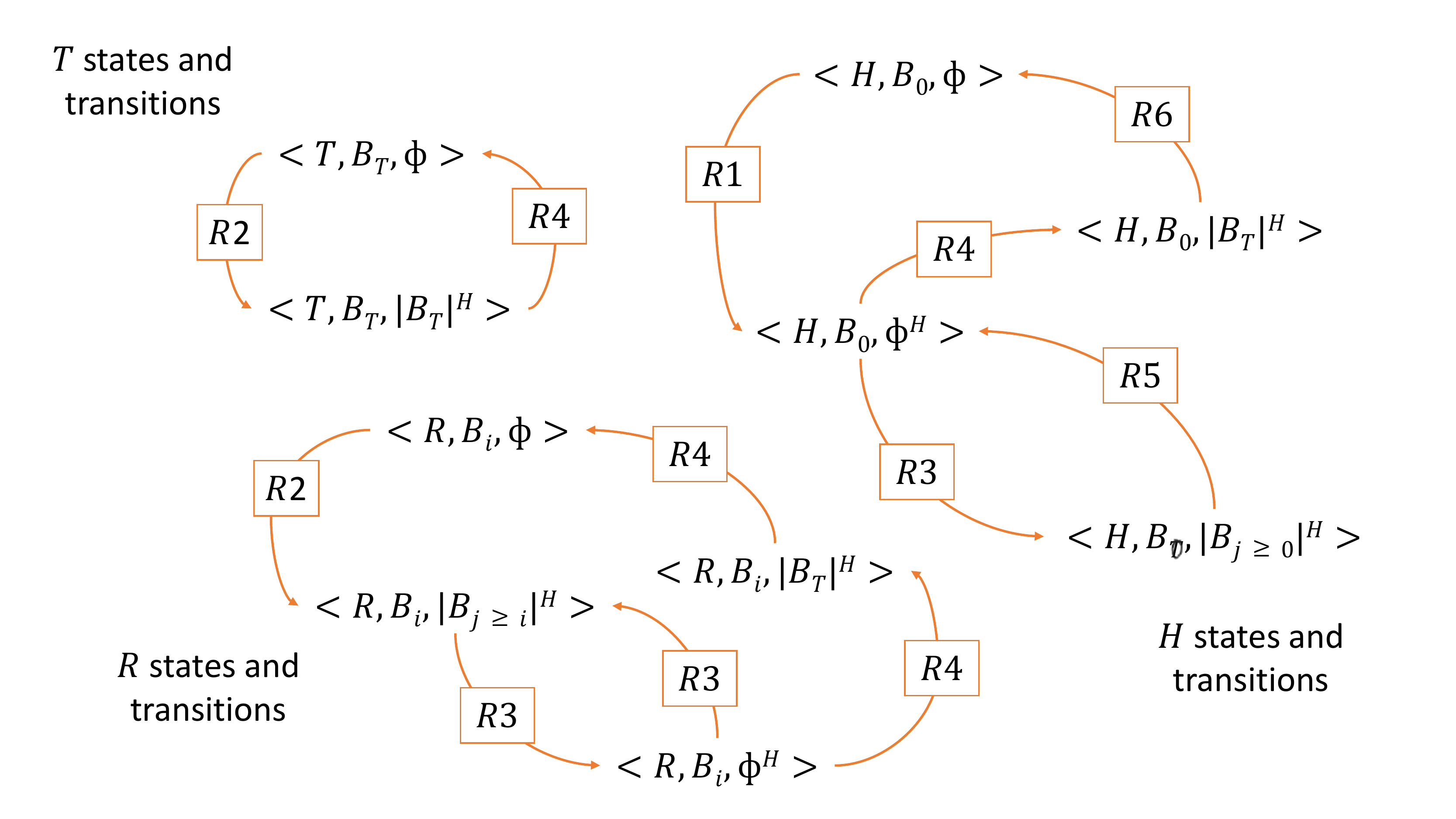}
    \caption{The old strand states and transitions}
    \label{f:old}
\end{figure}

\begin{figure}[ht]

    \centering
    \includegraphics[scale=0.35]{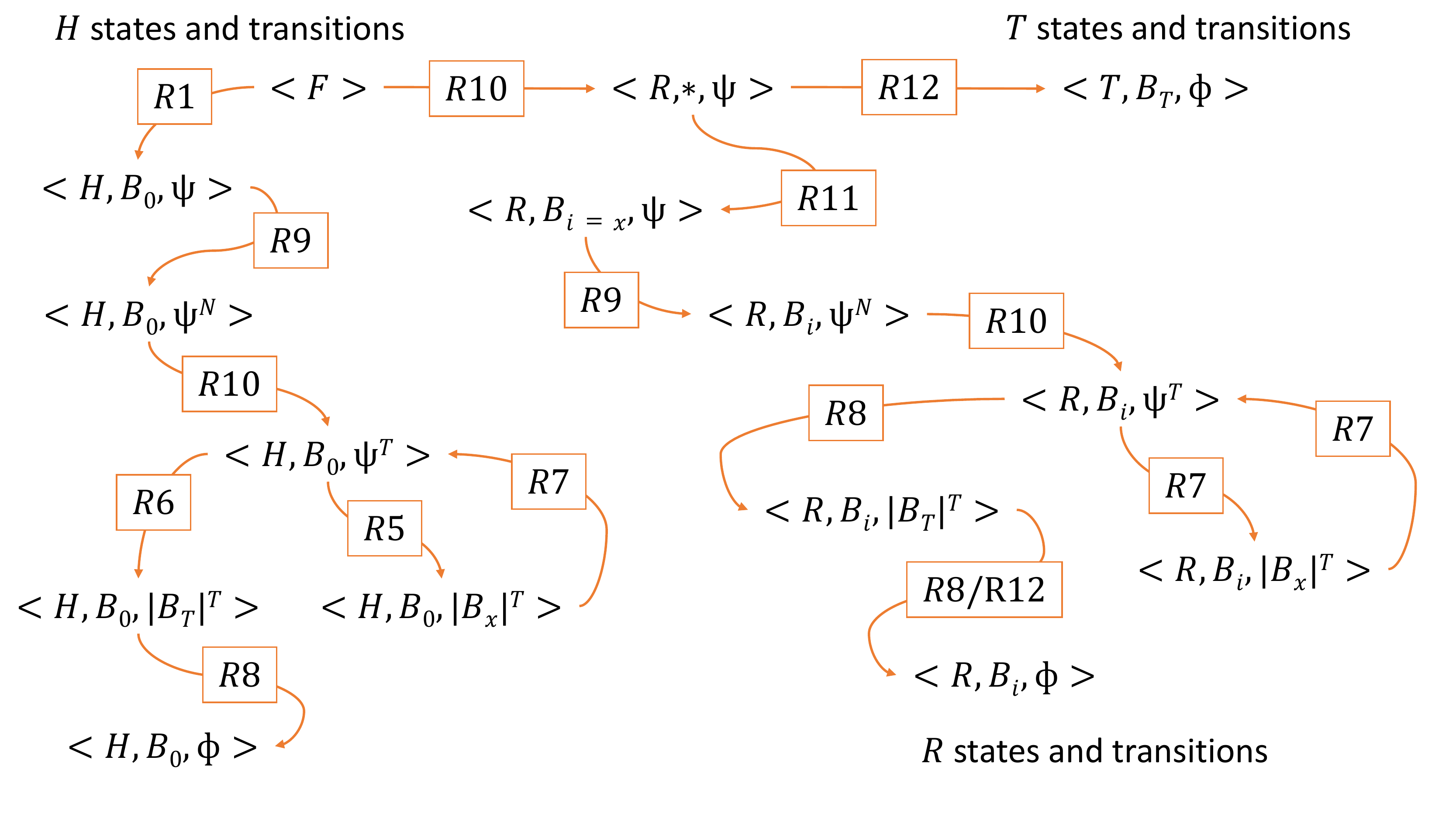}
    \caption{The new line states and transitions}
    \label{f:new}
\end{figure}

\end{proof}



\begin{lemma}\label{comp}
The strand self-replication protocol based on rules {\bf R1-R12} stabilises in parallel time $O(n(k+\log n))$ whp. 
\end{lemma}

\begin{proof}
The strand self-replication protocol mimics the pipelining mechanism 
utilised and analysed in the probabilistic bubble-sort procedure. 
The main differences is the fact that 
the bits of information are moved along the original and the new strand
at the same time as the new strand is being constructed.
In particular, when the leading bit reaches the current end of the new strand, 
the extension request (for a new agent and edge connection) is successful with probability $\le 1/\binom{n}{2}.$ 
Also move of any bit along an existing edge is successful with probability $1/\binom{n}{2}.$
Thus the expected potential change associated with one interaction 
(the counterpart of the inequality from Lemma~\ref{bubble}) is
\[EP(C')\le \left(1-\frac{1}{2n(n-1)}\right) P(C). 
\]
As the total number of extension requests is $k$ and 
the longest distance any bit has to move is $2k,$
%
%
we get the initial potential $P(C_0)\le 2^{3k}$ in this case.
In order to estimate the number of interactions $t$, after which
$EP(C_t)\le n^{-\eta}$, we get inequality
\[EP(C_t)\le \left(1-\frac{1}{2n(n-1)}\right)^t P(C_0)\le \exp\left(-\frac{t}{2n(n-1)}\right)2^{3k}\le n^{-\eta}, 
\]
which holds for $3k\ln 2+\eta\ln n\le \frac{t}{2n(n-1)}$
and in turn for
\[t\ge 2n(n-1)(3k\ln 2 +\eta\ln n).
\]
\end{proof}

\end{proof}

\begin{corollary}
The strand self-replication protocol generates $l$ copies of a $k$-bit strand in parallel time $O(n(k+\log n)\log l)$ whp. 
\end{corollary}

\subsection{Pattern matching with strand self-replication\label{patmat}}
{\em Pattern matching} is a classical problem in Algorithms~\cite{Crochemore2016}.
In this problem there are two strings,
a shorter {\em pattern} $P$ of length $k$ and usually much longer {\em text} $T$ of length $m$.
The main task in pattern matching is to find all occurrences of $P$ in $T$.
We demonstrate how to utilise self-replication mechanism in pattern matching
in the network constructors model.

Assume we have two strands, one containing $P^R$ 
(reversed sequence of bits in $P$) and the other containing $T.$
One can find all occurrences of $P$ in $T$ by 
forming a single strand containing sequence $T\cdot P^R$, 
further injection to and pipelining across $T$ the consecutive bits of $P^R$
while adopting the pattern matching procedure from~\cite{CG94}.
Using this approach and Theorem~\ref{comp} one can prove that 
utilising a fixed number of states the parallel time of finding 
all occurrences of $P$ in $T$ is $O(n(m+k+\log n))$ whp.  

The pattern matching protocol can be improved by 
instructing the original strand and all replicas containing $P^R$ to alternate 
between insertion of its content at a random position in $T$ and self-replication.
Each insertion and self-replication takes time $O(n(k+\log n)).$
After $\frac{m}{k}$ pattern replications and insertions the distance between 
any two consecutive insertion points in strand $T$ is $O(k\log m)$ whp, 
for $m$ large enough.
Thus the parallel time of the improved protocol is independent from $m$
and is bounded by $O(n(k+\log n)\log n),$ where 
$O(n(k+\log n)\log{\frac{m}{k}})$ comes from all insertions and self-replications, and
$O(n(k\log m+\log n))$ refers to pattern matching on each segment of size $O(k\log m).$

\section{Conclusion and Open Problems}
Our new $O(n\log n)$ parallel time spanning line construction protocol is optimal whp.
Please note that the lower bound argument for line formation from Theorem \ref{lowerbound} does not depend on the size of the state space. 
This means that $o(n\log n)$ parallel time line construction whp is not possible even if the state space is arbitrarily large.
Note also, that using an analogous formal argument on can prove $\Omega(n)$ parallel time lower bound for line construction 
in expectation, when the high probability guaranties are not imposed.
In fact, under such relaxed probabilistic requirements and slightly increased to $O(\log n)$ state space one can construct a spanning line in parallel time $O(n\log\log n)$ in expectation.
Without going into great detail,
this is done by initial simultaneous construction of $\log n$ (shorter) independent lines $L_1,\dots,L_{\log n}$ spanning all agents, followed by establishing connections between the relevant endpoints of lines $L_i$ and $L_{i+1},$ for $i=1,\dots,\log(n)-1.$ 
This is a separation result distinguishing between different probabilistic qualitative expectations. It also suggests a limited trade-off between the state space and the expected time for line formation.
Finally note, that the line formation protocol described in this paper can be extended to spanning ring formation with the help of the leader based phase clock. Such ring formation protocol stabilises in $O(n\log n)$ parallel time, both in expectation and whp.

Going beyond the proposed strand self-replication protocol one could investigate whether other network structures can self-replicate and at what cost. 
Also further studies on utilisation of strands (as carriers of information) in more complex distributed processes is needed.
Finally, one should seek alternative models for population protocols to  resolve
the bottleneck of infrequent 
(with probability $\Theta(\frac{1}{n^2})$) 
visits to the existing edges, which would likely result in faster construction and replication protocols. 

%

\bibliography{references}

\end{document}